\documentclass[pmlr]{jmlr} 

\usepackage{booktabs} 
\usepackage{graphicx} 
\usepackage{url}      
\usepackage{enumitem}
\usepackage{amsmath}  
\usepackage{fancyvrb}
\usepackage{algorithm2e} 
\usepackage{tikz}
\usepackage{mdframed}
\usetikzlibrary{arrows.meta, positioning}

\title[Semiparametric Inference for Half-Trek Estimators in Linear SEMs]{Semiparametric Inference for Half-Trek Estimators in Linear Structural Equation Models}

\author{\Name{Leopold Mareis} \Email{leopold.mareis@tum.de}\\
  \addr Technical University of Munich\\
  Munich, Germany
  \AND
  \Name{Nils Sturma} \Email{nils.sturma@epfl.ch}\\
  \addr École Polytechnique Fédérale de Lausanne \\ 
  Lausanne, Switzerland
  \AND
  \Name{Mathias Drton} \Email{mathias.drton@tum.de}\\
  \addr Munich Center for Machine Learning \\ 
  Technical University of Munich\\
  Munich, Germany
}

\newcommand\indep{\protect\mathpalette{\protect\independenT}{\perp}}
\def\independenT#1#2{\mathrel{\rlap{$#1#2$}\mkern3mu{#1#2}}}
\newcommand{\E}[1]{\mathrm{E}\!\left[ #1 \right]}

\newmdenv[
  linewidth=0.4pt,
  linecolor=black,
  innerleftmargin=8pt,
  innerrightmargin=8pt,
  innertopmargin=6pt,
  innerbottommargin=6pt,
  skipabove=\topsep,
  skipbelow=\topsep,
]{restatebox}
\newcommand{\restatement}[2]{%
  \begin{restatebox}
    \textbf{#1:} #2
  \end{restatebox}%
}

\begin{document}

\maketitle

\begin{abstract}
Linear structural equation models on directed mixed graphs encode causal relationships among variables subject to latent confounding.
The half-trek criterion (HTC) provides a graphical sufficient condition for the structural coefficients to be rationally identifiable from the observable covariance matrix, and yields a corresponding closed-form rational estimator.
Despite this, the asymptotic distribution of the HTC estimator, and hence valid standard errors and confidence regions, have not been derived.
We derive the semiparametric influence function of this estimator for every HTC-identified node, in arbitrary directed mixed graphs, including cyclic ones.
The influence function combines the structural residual at the target node with the identification instruments, recursively corrected for uncertainty from earlier estimation stages.
The HTC estimator is asymptotically normal with variance computable in closed form, yielding confidence regions, marginal intervals, and Wald tests for individual structural coefficients. 
Applied to the Fulton Fish Market dataset, our theory delivers a complete inferential summary for the causal effect of supply on demand.

\end{abstract}


\section{Introduction}\label{sec:intro}
Linear structural equation models (SEMs) on directed mixed graphs are a central tool in causal inference, valued for the direct interpretability of their structural coefficients and their ability to encode latent confounding through bidirected edges \citep{Bollen1989, Spirtes2000, drton2019}.
In such a model, a random vector $X = (X_v)_{v \in V}$ satisfies
\[
  X_v = \sum_{w \in \mathrm{pa}(v)} \beta_{wv} X_w + \varepsilon_v,
  \quad v \in V,
\]
according to a directed mixed graph $G = (V, D, B)$, where errors $\varepsilon_v$ and $\varepsilon_w$ may be dependent whenever $v \leftrightarrow w \in B$.
The fundamental question is whether $\beta$ is identifiable from the observed covariance matrix $\Sigma$; once it is, a natural second question arises: what is the asymptotic distribution of the resulting estimator, and how can it be used for inference?
This paper answers the second question for the half-trek criterion (HTC) estimator.

\paragraph{The half-trek criterion.}
The classical approach to identification under confounding is the instrumental variable (IV) estimator of \citet{Wright1928} and \citet{Bowden1984}.
In the graph $X_1 \to X_2 \to X_3$, $X_2 \leftrightarrow X_3$, the coefficient $\beta_{23}$ is unidentified by regression but recovered by $\beta_{23} = \E{X_1 X_3} / \E{X_1 X_2}$, because $X_1$ is correlated with $X_2$ yet uncorrelated with $\varepsilon_3$.
A sequence of graphical criteria for linear SEMs extended this idea to progressively richer confounding structures \citep{Brito2006, Tian2009, Foygel2012, Chen2016, Chen2017, Kumor2019, Barber2022}, each expressing identified parameters as rational functions of the covariance matrix $\Sigma$.
Among these, the HTC of \citet{Foygel2012} is the first polynomial-time constructive criterion applicable to arbitrary directed mixed graphs, including cyclic ones, without restrictions on the error-covariance structure: 
Whenever it succeeds at node $v$, it locates a \emph{witness set} $Y_v$ and identifies $\beta_v = (\beta_{wv})_{w \in \mathrm{pa}(v)}$, the structural coefficients of edges into $v$, by solving the IV system $A_v \beta_v = b_v$, where both $A_v$ and $b_v$ are computable from $\Sigma$. In particular, $A_v = \mathrm{E}[Z_{Y_v} X_{\mathrm{pa}(v)}^\top]$ is an instrument-relevance matrix, 
where, 
for each  $y \in Y_v$, the instrument $Z_y$ is either equal to the raw variable $X_y$,  
or equal to its structural residual $\varepsilon_y = X_y - \sum_{q \in \text{pa}(y)} \beta_{qy} X_q$ when $\beta_y$ has been identified at an earlier stage. The role of $y$ is determined by a graphical reachability condition on $G$ formalised in Section~\ref{sec:htc_estimation}.
This iterative construction gives an \emph{HTC ordering} $\prec$. While the HTC is not necessary for identifiability, complete criteria via Gr\"{o}bner basis computations exist \citep{Garcia2010, Hollering2026}. 

\paragraph{The inferential gap.}
Despite yielding a closed-form rational estimator, the HTC provides no asymptotic distribution, no standard errors, no confidence regions, and no such theory has been available.
Answering the inferential question requires two advances absent from classical single-stage IV: a joint matrix delta-method expansion handling all parents of $v$ simultaneously, and recursive correction terms carried through $\prec$ that account for estimation uncertainty in earlier stages $\hat{\beta}_y$.
We therefore address the following question: given HTC identification of $\beta_v$, what is the asymptotic distribution of $\hat{\beta}_v$?

\paragraph{Semiparametric inference for structural coefficients.}
Semiparametric efficiency for IV estimation was characterized by \citet{CHAMBERLAIN1987305} and \citet{Newey1990}, and efficient estimators for general nonlinear causal effects are studied in \citet{Jung2021}.
For linear confounded graphs, \citet{Mareis2026} derive the efficient influence function of the front-door estimator.
In the confounding-free setting, \citet{WitteHenckelMaathuis2020} and \citet{HenckelPerkovicMaathuis2022} characterize variance-optimal adjustment sets for total effect estimation, and \citet{HenckelButtenschoenMaathuis2024} extend this to acyclic directed mixed models.
For which directed mixed graphs the semiparametric efficiency bound is attainable within the HTC framework remains open: in the graph $1 \to 3 \to 4$, $2 \to 3 \leftrightarrow 4$, the bound for $\beta_{34}$ is achieved by an estimator using both instruments simultaneously, a bound no single-witness HTC estimator attains.
We therefore establish asymptotic normality of the HTC estimator as a foundation for inference, working within the semiparametric framework of \citet{Vaart_1998} and \citet{Tsiatis_Semiparametric} and imposing only finite fourth moments and differentiability in quadratic mean at the truth.

\paragraph{Contributions.}
We establish three results.
\begin{enumerate}[label=(\roman*), itemsep=3pt, parsep=0pt, leftmargin=2em]
  \item \textbf{Influence function.} For any HTC-identified node v in an arbitrary directed mixed graph, including cyclic ones, the semiparametric influence function of $\hat{\beta}_v$ is $\phi_{\beta_v} = A_v^{-1} R_v$, where the vector $R_v$ collects the row residuals of the identification system with recursive corrections from earlier estimation stages (Theorem~\ref{thm:htc_if}).
  \item \textbf{Asymptotic normality.} The asymptotic variance
    $\mathcal{V}_v = A_v^{-1} \E{R_v R_v^\top} A_v^{-\top}$ is expressed via finitely many cross-moments computable by descent through $\prec$, and reduces to the standard 2SLS variance for the classical IV graph (Proposition~\ref{prop:recursive_variance}).
  \item \textbf{Inference.} Confidence regions, marginal intervals, and Wald tests are provided, implemented in \texttt{R} with output mirroring \texttt{summary.lm()} \citep{R_stats}. Applied to the Fulton Fish Market \citep{Graddy1995, Angrist2000}, the theory delivers a complete inferential summary for the causal effect of supply on demand in a simultaneous-equations model with correlated errors.
\end{enumerate}

\paragraph{Organisation.}
Section~\ref{sec:htc_estimation} introduces the statistical model and the HTC estimator.
Section~\ref{sec:if} derives $\phi_{\beta_v}$, establishes asymptotic normality, and provides the closed-form asymptotic variance $\mathcal{V}_v$.
Section~\ref{sec:numerical_experiments} reports simulation evidence for calibration of the estimator $\hat{\mathcal{V}}_v$ across Gaussian and non-Gaussian errors, and examines the effect of witness set choice on estimation variance.
Section~\ref{sec:inference} translates the limit theory into valid inference, presents the \texttt{R} implementation, and illustrates the complete workflow on the Fulton Fish Market dataset.
Regularity conditions are stated in Appendix~\ref{asm:regularity} and are in force throughout.

\section{HTC Identification and Estimation}\label{sec:htc_estimation}
\subsection{Directed Mixed Graphs and Half-Treks}
Let $G=(V,D,B)$ be a directed mixed graph on the finite vertex set $V$ with directed edge set $D\subseteq V^{2}$ and bidirected edge set $B\subseteq \binom{V}{2}$.  
We also denote the directed edges $(v,w) \in D$ by $v \rightarrow w$ and the bidirected edges $\{v,w\} \in B$ by $v \leftrightarrow w$. Moreover, we assume that neither the directed part nor the bidirected part contain self-loops, that is, $v \to v \not\in D$ and $v \leftrightarrow v \not\in B$ for all $v \in V$.
Define the \emph{parents} of node $v$ as $\mathrm{pa}(v)=\{w\in V:w\to v\in D\}$ and its \emph{siblings} as $\mathrm{sib}(v)=\{w\in V:w\leftrightarrow v\in B\}$. 
We extend the latter to sets $ W \subseteq V$ via $\mathrm{sib}(W) = \{v \in V \setminus W : v \leftrightarrow w \in B \text{ for some } w \in W\}$.
A \emph{half-trek} is a path of the form
\[
  v_0\to v_1\to\cdots\to v_\ell
  \qquad\text{or}\qquad
  v_0\leftrightarrow v_1\to\cdots\to v_\ell,
\]
with $\mathrm{Left}(\pi)=\{v_0\}$ and $\mathrm{Right}(\pi)=\{v_1,\ldots,v_\ell\}$.
The \emph{half-trek reachable set} of $v$ is
\[
  \mathrm{htr}(v)
  :=
  \bigl\{w\in V\setminus(\{v\}\cup\mathrm{sib}(v)):
         \text{there exists a half-trek from }v\text{ to }w\bigr\}.
\]
A system of half-treks has \emph{no sided intersection} if $\mathrm{Left}(\pi_i)\cap\mathrm{Left}(\pi_j)=\emptyset =\mathrm{Right}(\pi_i)\cap\mathrm{Right}(\pi_j)$ for all $i\neq j$.

\subsection{The Semiparametric Model}
The parameter matrix $\beta$ is the finite-dimensional parameter of interest, while the error distribution $\varepsilon$ is infinite-dimensional nuisance, making the model semiparametric.
We write $ L^p_0$ for the space of mean-zero $\mathbb{R}^{|V|}$-valued random vectors with finite $p$th moment, so $h\in L_0^p$ if $\E{h} = 0$ and $\E{\|h\|^p} < \infty$.
\begin{definition}[Admissible parameters and nuisance]
\label{def:parameter_space} 
Let $G=(V,D,B)$ be a directed mixed graph. 
The set of admissible parameter matrices is
\[\mathcal{B}=\{\beta\in\mathbb{R}^{|V|\times |V|}:\beta_{wv}=0\text{ whenever }w\to v\notin D,\text{ and } \det(I-\beta) \neq 0 \}.\]
For acyclic graphs, $\det(I-\beta)\neq 0$ holds automatically; for cyclic graphs it is a non-trivial constraint.
The nuisance space $\mathcal{E}$ collects random vectors satisfying the \emph{connected set Markov property}, also called the \emph{marginal independence model}:
\[
  \mathcal{E}
  =
    \left\{\varepsilon \in L^4_0 \, \middle| \, 
\varepsilon_W\perp\!\!\!\perp\varepsilon_{V\setminus(W\cup \, \mathrm{sib}(W))}
    \text{ for every }W\subseteq V\text{ connected in }(V,B)
\right\}.
\]
\end{definition}
\begin{definition}[Linear model]
\label{def:model}
The \emph{linear model} of a directed mixed graph $G$ is the set of probability distributions
$
  \mathcal{M}_G
  =
  \bigl\{ P_{(\beta,\varepsilon)}
         :\beta\in\mathcal{B},\;\varepsilon\in\mathcal{E}\bigr\},
$
where $P_{(\beta,\varepsilon)}$ denotes the distribution of a random vector $X$ solving the equation system $X=\beta^{\top}X+\varepsilon$, that is, $X=(I-\beta)^{-\top}\varepsilon$.
\end{definition}
The estimation target is the parameter $\beta\in\mathcal{B}$.
Based on i.i.d.~samples $(X^{(i)})_{i \in [n]}$ of an unknown distribution  $P_0 \in \mathcal{M}_G$, a \emph{regular asymptotically linear} (RAL) estimator $\hat\beta_n$ satisfies
\[
  \sqrt{n}(\hat\beta_n-\beta)
  =
  \frac{1}{\sqrt{n}}\sum_{i=1}^{n}\phi_\beta(X^{(i)})+o_{P_0}(1)
  \;\xrightarrow{d}\;
  \mathcal{N}\!\left(0,\,\E{\phi_\beta\phi_\beta^\top}\right)
\]
for some mean-zero, square-integrable \emph{influence function} $\phi_\beta$.
For any parametric path $\gamma\mapsto P_\gamma\in\mathcal{M}_G$ through $P_0$ that is differentiable in quadratic mean at $\gamma = 0$ (Assumption~\ref{asm:regularity}.\ref{asm:1}), with score $S\in L^2_0$ as in \ref{asm:regularity}.\ref{asm:1} and induced parameter curve $\beta_\gamma\in\mathcal{B}$, the influence function satisfies $\mathbb{E}[\phi_\beta S]=\frac{d}{d\gamma}\big|_{\gamma=0}\beta_\gamma$ for every such submodel \citep[§25.3]{Vaart_1998}. 
Differentiating any population identifying equation for $\beta$ along $\gamma\mapsto P_\gamma$ recovers $\phi_\beta$.
In Sections~\ref{sec:htc}--\ref{sec:variance} we derive 
the influence function of the HTC estimator of $\beta_v$.

\subsection{Identification and Instruments}\label{sec:htc}

\begin{definition}[Half-trek criterion]
\label{def:htc}
A set $Y\subset V$ satisfies the \emph{half-trek criterion} (HTC) for
$v\in V$ if
\begin{enumerate}[label=(\roman*), itemsep=3pt, parsep=0pt, leftmargin=2em]
\item $|Y|=|\mathrm{pa}(v)|$,
\item $Y\cap(\{v\}\cup\mathrm{sib}(v))=\emptyset$, and
\item there exists a system of half-treks from $Y$ to $\mathrm{pa}(v)$ with
      no sided intersections.
\end{enumerate}
A set $Y\subset V$ satisfying the HTC for $v$ is called a \emph{witness set} for $v$; an element $y\in Y$ is an \emph{internal witness} if $y\in\mathrm{htr}(v)$ and an \emph{external witness} if $y\notin\mathrm{htr}(v)$.
\end{definition}
\citet[Theorem~1]{Foygel2012} showed that the parameters $\beta_{W,\cdot}$ are \emph{HTC-identified} in $G$ for $W \subseteq V$ if a family $(Y_v:v\in W)$ and a total order $\prec$ on $W$ exist such that $Y_v$ satisfies the HTC for every $v$ and $w\prec v$ whenever $w\in Y_v\cap\mathrm{htr}(v)$. We call $\prec$ the \emph{HTC ordering}; nodes with $Y_v\cap\mathrm{htr}(v)=\emptyset$ are minimal in $\prec$ and serve as the base of the induction in Section~\ref{sec:if}.
To identify the parameter vector $\beta_v:=(\beta_{p_1v},\ldots,\beta_{p_kv})^{\top}$ of directed edges pointing into $ v $, a raw HTC instrument vector $Z_{Y_v} := (Z_{y_1}, \dots , Z_{y_k})^\top$ is component-wise constructed through 
\[
  Z_y
  :=
  \begin{cases}
    \varepsilon_y = X_y - \beta_y^{\top}X_{\mathrm{pa}(y)}, & y\in\mathrm{htr}(v),\\
    X_y, & y\notin\mathrm{htr}(v).
  \end{cases} 
\]
Each $\beta_y$ was identified before since $y \in Y_v \cap \mathrm{htr}(v)$. The \emph{HTC relevance matrix} $A_v:=\E{Z_{Y_v}X_{\mathrm{pa}(v)}^{\top}}$ coincides with the identification matrix studied by \citet{Foygel2012}:
\newcommand{\TextLemmaHTCRelevanceMatrix}{
The HTC relevance matrix $A_v$ depends on $(\beta, \varepsilon)$ only through
$(\beta, \Omega)$ where $\Omega := \mathrm{Cov}(\varepsilon)$ is in  $\mathrm{PD}(V)$, the cone of positive definite matrices.
It is generically invertible: the set of $(\beta, \Omega) \in \mathcal{B} \times
\mathrm{PD}(V)$ for which $A_v$ is not invertible has Lebesgue measure zero.
}
\begin{lemma}[HTC relevance matrix]
\label{lem:htc_matrix}
\TextLemmaHTCRelevanceMatrix
\end{lemma}
All proofs are presented in Appendix \ref{appendix_proofs}. 
Lemma~\ref{lem:htc_matrix} shows invertibility holds generically; we assume it holds at the true parameter throughout (Assumption~\ref{asm:regularity}.\ref{asm:2}).
Setting $b_v:=\E{Z_{Y_v}X_v}$, 
the rotated instruments $A_v^{-1}Z_{Y_v}$ satisfy the standard IV orthogonality conditions:
\[
  \E{A_v^{-1}Z_{Y_v}X_{\mathrm{pa}(v)}^\top} = I_{|\mathrm{pa}(v)|},
  \qquad
  \E{A_v^{-1}Z_{Y_v}X_v} = A_v^{-1}b_v.
\]
The first identity shows that the $j$th row $e_j^\top A_v^{-1}Z_{Y_v}$ is uncorrelated with all parents except $p_j$; the second recovers $\beta_v=A_v^{-1}b_v$ as a population moment via the HTC orthogonality $\E{Z_y\varepsilon_v}=0$, the central equation 
Section~\ref{sec:if} differentiates to yield the influence function $\phi_{\beta_v}$.

\section{Influence Function and Asymptotic Variance}\label{sec:if}
Write $\hat\beta_{\prec v} := (\hat\beta_y)_{y \prec v}$ for the estimates of nodes preceding $v$. 
The empirical HTC estimator $\hat\beta_v := \hat{A}_v^{-1}\hat{b}_v$ is a Z-estimator: at each $v$ in HTC order it solves
\[
  n^{-1}\sum_{i \in [n]}\hat Z_{Y_v}(X^{(i)},\hat\beta_{\prec v})\, \left( X_v^{(i)} - \hat\beta_v^{\top}X_{\mathrm{pa}(v)}^{(i)} \right) = 0,
\]
with $\hat Z_y(X, \hat\beta_{\prec v})=X_y-\hat\beta_y^{\top}X_{\mathrm{pa}(y)}$ for $y\in\mathrm{htr}(v)$ and $\hat Z_y = X_y$ otherwise.
The identification equation $A_v\beta_v=b_v$ has rows indexed by witnesses $y\in Y_v$; the influence function $\phi_{\beta_v}=\phi_{A_v^{-1}b_v}$ follows by the delta method, which requires the pathwise derivatives $\phi_{M_y(t)}$ of the individual moments $M_y(t):=\E{Z_yX_t}$.
For external witnesses, the instrument $Z_y=X_y$ is fixed and $\phi_{M_y(t)}$ is the standard covariance influence function.
For internal witnesses, the residual instrument $Z_y=X_y-\beta_y^{\top}X_{\mathrm{pa}(y)}$ depends on $\beta_y$, so the pathwise derivative carries a correction from the earlier influence functions $\phi_{\beta_{qy}}$.
Lemma~\ref{lem:row_if} collects both cases; Lemma~\ref{lem:explicit_row} combines them into the row contributions $R_{y,v}$ entering Theorem~\ref{thm:htc_if} for $\phi_{\beta_v}=A_v^{-1}R_v$.

\newcommand{\TextRowMomentDerivative}{
For $M_y(t):=\E{Z_yX_t}$ with $y\in Y_v$, $t\in V$, the influence function along any parametric submodel is
\[
  \phi_{M_y(t)}
  =
  \begin{cases}
    X_yX_t - \Sigma_{yt}, & y\notin\mathrm{htr}(v),\\[4pt]
    \varepsilon_yX_t - \E{\varepsilon_yX_t}
    -\displaystyle\sum_{q\in\mathrm{pa}(y)}\Sigma_{qt}\,\phi_{\beta_{qy}},
    & y\in\mathrm{htr}(v).
  \end{cases}
\]}
\begin{lemma}[Row moment derivative]
\label{lem:row_if}
\TextRowMomentDerivative
\end{lemma}
The row contribution $R_{y,v}$ is the influence function of the $y$-th identification equation residual $e_y^\top (b_v - A_v \beta_v) = M_y(v)-\sum_{p\in\mathrm{pa}(v)}\beta_{pv}M_y(p)$. It is obtained by applying Lemma~\ref{lem:row_if} to $M_y(v)$ and each $M_y(p)$ and combining linearly.

\newcommand{\TexRowContribution}{
For witness $y\in Y_v$ and target $v$, the influence function of the $y$-th row of the identification equation
\[
  R_{y,v}:=\phi_{M_y(v)}-\sum_{p\in\mathrm{pa}(v)}\beta_{pv}\,\phi_{M_y(p)},
\]
 satisfies
\[
  R_{y,v}
  =
  \begin{cases}
    X_y\varepsilon_v, & y\notin\mathrm{htr}(v),\\[6pt]
    \varepsilon_y\varepsilon_v
    -\displaystyle\sum_{q\in\mathrm{pa}(y)}\E{X_q\varepsilon_v}\,\phi_{\beta_{qy}},
    & y\in\mathrm{htr}(v).
  \end{cases}
\]
Write $R_v:=(R_{y_1,v},\ldots,R_{y_k,v})^{\top}$.}
\begin{lemma}[Row contribution]
\label{lem:explicit_row}
\TexRowContribution
\end{lemma}

\subsection{The Influence Function $\phi_{\beta_v}$}

The influence function for $\beta_v=A_v^{-1}b_v$ follows by the delta method applied to the identification equation, with Lemmas~\ref{lem:row_if} and~\ref{lem:explicit_row} supplying the row-wise derivatives.

\newcommand{\TextHTCInfluenceFunction}{ 
The functional $\phi_{\beta_v}=A_v^{-1}R_v$ is an influence function for $\beta_v$ in $\mathcal{M}_G$.
For the edge $p_j\to v$, this yields $\phi_{\beta_{p_jv}}=e_j^{\top}A_v^{-1}R_v$.}
\begin{theorem}[HTC influence function]
\label{thm:htc_if}
\TextHTCInfluenceFunction
\end{theorem}

\begin{example}
In the classical IV case with $|\mathrm{pa}(v)|=1$ and one external witness $y\notin\mathrm{htr}(v)$, Lemma~\ref{lem:explicit_row} gives $R_{y,v}=X_y\varepsilon_v$ and $A_v=\E{X_yX_p}$, so the HTC influence function $\phi_{\beta_{pv}}$ in Theorem \ref{thm:htc_if} evaluates to $ X_y\varepsilon_v/\E{X_yX_p}$, the standard 2SLS influence function.
\end{example}

\begin{remark}[Semiparametric efficiency of HTC estimators]
\label{rem:efficiency}
The influence function $\phi_{\beta_v}=A_v^{-1}R_v$ of Theorem~\ref{thm:htc_if} is the influence function of the HTC estimator for the chosen witness set $Y_v$, not necessarily the efficient influence function of $\beta_v$ in $\mathcal{M}_G$.
In the graph $1\to 3\to 4$, $2\to 3\leftrightarrow 4$, the efficient influence function for $\beta_{34}$ uses both instruments simultaneously \citep{CHAMBERLAIN1987305,Newey1990}, a bound no single-witness HTC estimator achieves.
\end{remark}

\subsection{The Asymptotic Variance $\mathcal{V}_v$}\label{sec:variance}

The asymptotic normality of the HTC estimator, written $\hat\beta_{v,n}$ to make the sample-size dependence explicit, follows from Theorem~\ref{thm:htc_if} by the central limit theorem: since $\phi_{\beta_v}=A_v^{-1}R_v$ is mean-zero and square-integrable, the sandwich formula below is its asymptotic~variance.

\newcommand{\TextPropRecursiveVariance}{
The HTC estimator $\hat\beta_{v,n}$ satisfies
\[
\sqrt{n}(\hat\beta_{v,n}-\beta_v)\xrightarrow{d}\mathcal{N}(0,\mathcal{V}_v) \quad \text{with} \quad
  \mathcal{V}_v
  =
  A_v^{-1}\,\E{R_vR_v^{\top}}\,A_v^{-\top}.
\]
For the edge $p_j\to v$, $\mathcal{V}_{v}[j,j] = e_j^{\top}\mathcal{V}_v e_j$.
The covariance $\E{R_vR_v^{\top}}$ expands for $a,b\in Y_v$ to:\\
\textbf{(i)} $a,b\notin\mathrm{htr}(v)$:
\[
  \E{R_{a,v}R_{b,v}}
  =
  \E{X_aX_b\varepsilon_v^2}.
\]
Both witnesses are external, so $R_{a,v}=X_a\varepsilon_v$ and $R_{b,v}=X_b\varepsilon_v$ with no correction.\\
\textbf{(ii)} $a\in\mathrm{htr}(v)$, $b\notin\mathrm{htr}(v)$:
\[
  \E{R_{a,v}R_{b,v}}
  =
  \E{\varepsilon_aX_b\varepsilon_v^2}
  -
  \sum_{q\in\mathrm{pa}(a)}
  \E{X_q\varepsilon_v}\,
  \E{\phi_{\beta_{qa}}X_b\varepsilon_v}.
\]
The subtracted term corrects for the uncertainty of $\beta_{qa}$: estimating the residual instrument $Z_a=\varepsilon_a$ requires $\hat\beta_{qa}$, which introduces a correction proportional to $\phi_{\beta_{qa}}$ into~$R_{a,v}$.\\
\textbf{(iii)} $a,b\in\mathrm{htr}(v)$:
\begin{align*}
  \E{R_{a,v}R_{b,v}}
  &=
  \E{\varepsilon_a\varepsilon_b\varepsilon_v^2}
  -\sum_{r\in\mathrm{pa}(b)}
  \E{X_r\varepsilon_v}\,\E{\varepsilon_a\varepsilon_v\phi_{\beta_{rb}}}
  -\sum_{q\in\mathrm{pa}(a)}
  \E{X_q\varepsilon_v}\,\E{\phi_{\beta_{qa}}\varepsilon_b\varepsilon_v} \\
  &\quad
  +\sum_{q\in\mathrm{pa}(a)}\sum_{r\in\mathrm{pa}(b)}
  \E{X_q\varepsilon_v}\,\E{X_r\varepsilon_v}\,
  \E{\phi_{\beta_{qa}}\phi_{\beta_{rb}}}.
\end{align*}
With both witnesses internal, first-stage corrections from $\phi_{\beta_{qa}}$ and $\phi_{\beta_{rb}}$ enter both factors; the double sum captures the covariance between these two correction terms.}
\begin{proposition}[Recursive variance formula] \label{prop:recursive_variance} 
\TextPropRecursiveVariance
\end{proposition}
The correction terms $\phi_{\beta_{qy}}$ in $R_v$ in Lemma \ref{lem:explicit_row} account for the full variation from estimating $\hat\beta_{\prec v}$, ensuring that the stated asymptotic distribution of $\hat\beta_v$ is valid at each estimation stage.

Proposition~\ref{prop:recursive_variance} involves cross-variances $\E{\phi_{\beta_{qa}}X_b\varepsilon_v}$ of earlier influence functions, giving $\mathcal{V}_v$ an apparently recursive definition. 
In Figure~\ref{fig:G1}, however, the recursion for $v=5$ bottoms out at node~$2$, which has $Y_2\cap\mathrm{htr}(2)=\emptyset$; there $\phi_{\beta_{12}}=X_1\varepsilon_2/\E{X_1^2}$ and $\E{\phi_{\beta_{12}}X_b\varepsilon_v}=\E{X_1\varepsilon_2 X_b \varepsilon_v} / \E{X_1^2}$ is a raw 2SLS fourth-order moment.

\newcommand{\TextLemmaVariancerecursion}{The covariance $\E{R_vR_v^{\top}}$ reduces to finitely many fourth-order moments of observables.
Each cross-variance $\E{\phi_{\beta_{qa}}X_b\varepsilon_v}$ equals
\[
  \E{\phi_{\beta_{qa}}X_b\varepsilon_v}
  =
  e_j^{\top}A_a^{-1}\,\E{R_aX_b\varepsilon_v},
\]
and for each $y\in Y_a$,
\[
  \E{R_{y,a}X_b\varepsilon_v}
  =
  \begin{cases}
    \E{X_y\varepsilon_aX_b\varepsilon_v} & y\notin\mathrm{htr}(a),\\
    \E{\varepsilon_y\varepsilon_aX_b\varepsilon_v}
    -\displaystyle\sum_{q'\in\mathrm{pa}(y)}\E{X_{q'}\varepsilon_a}\,
    \E{\phi_{\beta_{q'y}}X_b\varepsilon_v} & y\in\mathrm{htr}(a).
  \end{cases}
\]
Each application of the $y\in\mathrm{htr}(a)$ branch strictly descends in the HTC order $\prec$; the recursion bottoms out at nodes $u$ with $Y_u\cap\mathrm{htr}(u)=\emptyset$, where every entry equals a pure fourth-order moment of observables.}
\begin{lemma}[Termination and computability]
\label{lem:variance_recursion}
\TextLemmaVariancerecursion
\end{lemma}
Under Gaussian errors, Isserlis' theorem expresses these fourth-order moments as polynomials in $(\Sigma,\Omega,\beta)$, making $\mathcal{V}_v$ an explicit rational function of the model parameters.

The estimator $\hat{\mathcal{V}}_v$ is obtained by replacing every expectation in Proposition~\ref{prop:recursive_variance} and Lemma~\ref{lem:variance_recursion} with the corresponding empirical moment, evaluated on the observed sample.

\section{Numerical Experiments} \label{sec:numerical_experiments}
The simulation data aligned with the graph in Figure \ref{fig:G1} are generated by Gaussian and modified Gamma error distributions. Full analysis details are provided in Appendix~\ref{app:calibration}. 
\subsection{Calibration}
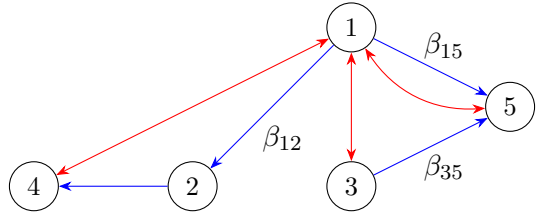
\begin{figure} \centering
\begin{tikzpicture}[scale = 0.7, >=Stealth, node style/.style={scale = 1, font=\small}]
\node[node style, draw, circle] (A) at (1.5,3) {$1$};
\node[node style, draw, circle] (B) at (-1.5,0) {$2$};
\node[node style, draw, circle] (C) at (1.5,0) {$3$};
\node[node style, draw, circle] (D) at (-4.5,0) {$4$};
\node[node style, draw, circle] (E) at (4.5,1.5) {$5$};
\node[black, scale = 1] (F) at (3.25,2.65) {$\beta_{15}$};
\node[black, scale = 1] (G) at (3.25,0.35) {$\beta_{35}$};
\node[black, scale = 1] (H) at (0.2,0.9) {$\beta_{12}$};

\def\dashlength{3pt}
\draw[blue, ->] (A) -- (B);
\draw[blue, ->] (B) -- (D);
\draw[blue, ->] (C) -- (E);
\draw[blue, ->] (A) -- (E);
\draw[red, <->] (A) -- (C);
\draw[red, <->] (A) -- (D);
\draw[red, <->] (A) to [bend right=30] (E);
\end{tikzpicture}
\caption{Five-node directed mixed graph (cf. Example 6(b) in \cite{Foygel2012}) with HTC ordering $2\prec 4\prec 5$.}
\label{fig:G1}
\end{figure}
\begin{figure} 
\centering
\includegraphics[width = 0.9\textwidth, trim=0 23pt 0 14pt, clip]{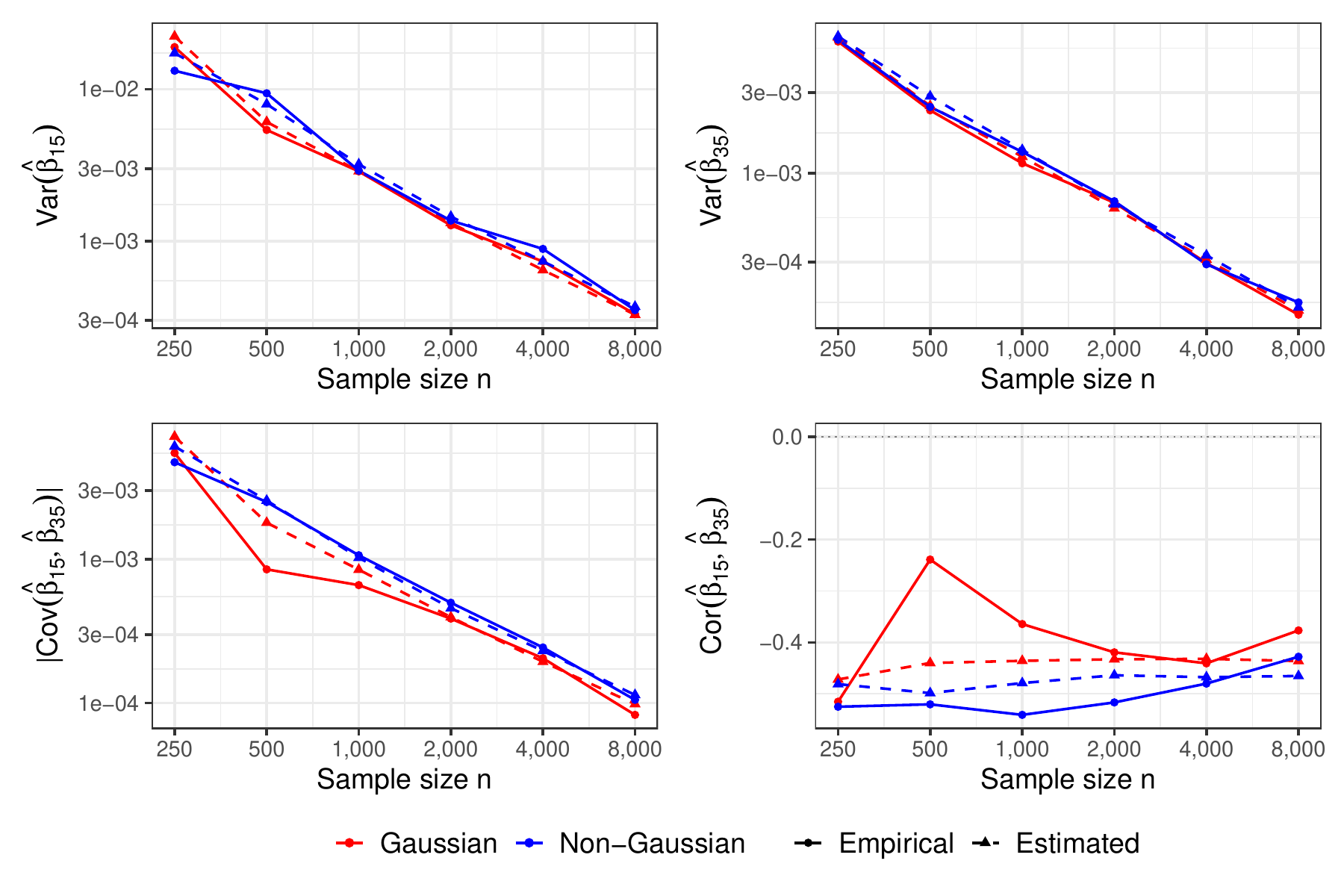}
\caption{Calibration of $\hat{\mathcal{V}}_5$ for $(\hat\beta_{15},\hat\beta_{35})$ in Figure~\ref{fig:G1} ($Y_5=\{3,4\}$, internal witness $4\in\mathrm{htr}(5)$, external witness $3\notin\mathrm{htr}(5)$): empirical variances, covariances, and correlations over $100$ replications (solid) versus mean $\hat{\mathcal{V}}_5/n$ (dashed), across sample sizes $n$ under Gaussian (red) and non-Gaussian (blue) errors.}
\label{fig:calibration}
\end{figure}
Figure~\ref{fig:calibration} shows that the asymptotic variance estimator $\hat{\mathcal{V}}_v$ is well calibrated under both Gaussian and non-Gaussian errors.
Since $\hat\beta_{15}$ and $\hat\beta_{35}$ are negatively correlated, the bottom-left panel displays the absolute covariance.

\subsection{Witness Choice and Estimation Variance}
\begin{figure}
\centering
\includegraphics[width = \textwidth, trim=0 20pt 0 10pt, clip]{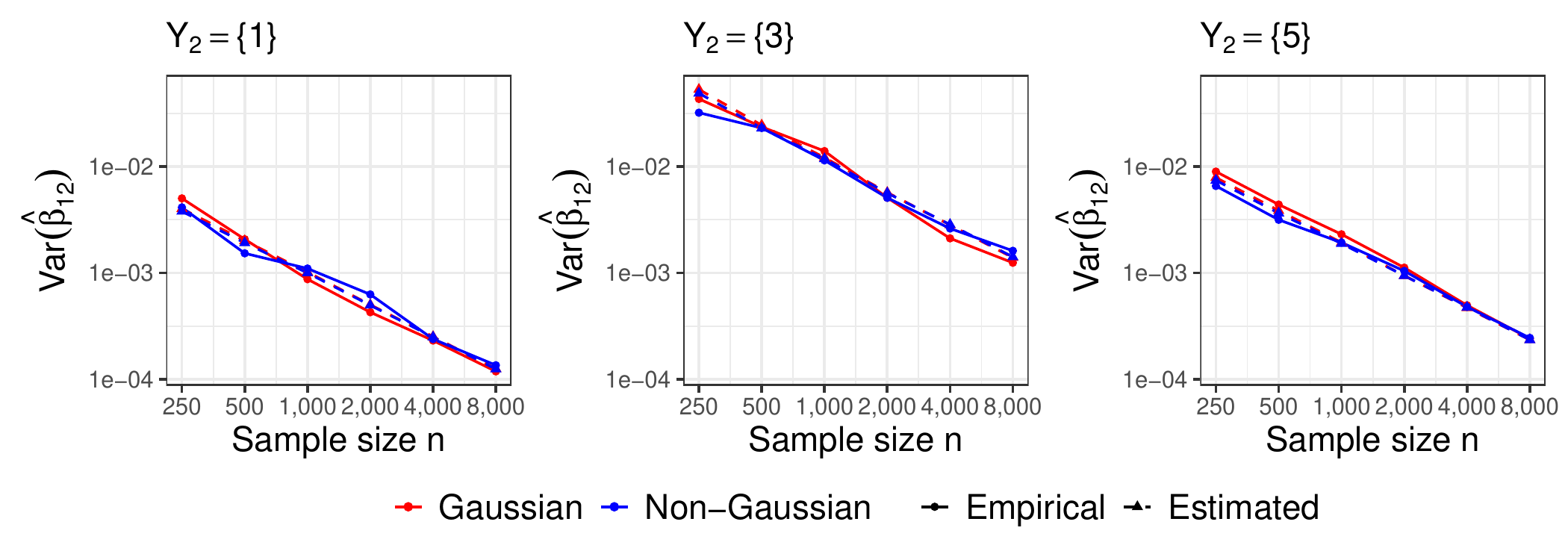}
\caption{Estimation variance of $\hat\beta_{12}$ in Figure~\ref{fig:G1} for three witness choices $Y_2\in\bigl\{\{1\},\allowbreak\{3\},\allowbreak\{5\}\bigr\}$: empirical variance over $100$ replications (solid) versus mean $\hat{\mathcal{V}}_2/n$ (dashed), across sample sizes $n$ under Gaussian (red) and non-Gaussian (blue) errors.}
\label{fig:efficiency}
\end{figure}

The witness set choice has an effect on estimation efficiency: Figure~\ref{fig:efficiency} shows that the asymptotic variance of $\hat\beta_{12}$ varies substantially across the three valid witness choices $Y_2\in\bigl\{\{1\},\{3\},\{5\}\bigr\}$.
In practice, the analyst must select a witness set, and a lower-variance choice is preferable.

\begin{remark}[Witness selection]
\label{rem:witness_selection}
The asymptotic variance estimator $\hat{\mathcal{V}}_v$ of Proposition~\ref{prop:recursive_variance} provides, for any fixed $Y_v$, a consistent estimate of $\mathcal{V}_v(Y_v)$, enabling post-hoc comparison of competing witness sets.
Selecting a variance-minimizing set in advance is an open problem: while valid HTC witness sets are combinatorially characterized, their complete enumeration requires searching over exponentially many candidate subsets of $V$, so no algorithmic guidance towards a lower-variance choice without exhaustive search is currently available.
\end{remark}

\section{Statistical Inference}\label{sec:inference}

Confidence regions and Wald tests for $\beta_v$ follow from the $\mathcal{N}(0,\mathcal{V}_v)$ limit of Proposition~\ref{prop:recursive_variance} via the asymptotic variance estimator $\hat{\mathcal{V}}_v$.

\newcommand{\TextPropInference}{
For $\alpha\in(0,1)$, write $z_{\alpha/2}$ for the $\alpha/2$ quantile of the standard normal distribution and $\chi^2_{k,1-\alpha}$ for the $(1-\alpha)$ quantile of the chi-squared distribution with $k$ degrees of freedom $\chi^2_k$. The following hold.
\begin{enumerate}[label=(\roman*), itemsep=3pt, parsep=0pt, leftmargin=2em]
\item \textit{(Confidence ellipsoid.)}
      An asymptotic $(1-\alpha)$ confidence region for $\beta_v\in\mathbb{R}^{|\mathrm{pa}(v)|}$ is
      \[
        \mathcal{C}_n
        =
        \bigl\{
          \beta\in\mathbb{R}^{|\mathrm{pa}(v)|}:
          n(\hat\beta_v-\beta)^{\top}\hat{\mathcal{V}}_v^{-1}(\hat\beta_v-\beta)
          \leq\chi^2_{|pa(v)|,1-\alpha}
        \bigr\}.
      \]
\item \textit{(Marginal intervals.)}
      An asymptotic $(1-\alpha)$ confidence interval for the edge $\beta_{p_jv}$ is
      \[
        \hat\beta_{p_jv}\pm z_{\alpha/2}\sqrt{\hat{\mathcal{V}}_v[j,j]/n}.
      \]
\item \textit{(Wald test.)}
      For $H_0\colon C\beta_v=c$ with full row rank $C\in\mathbb{R}^{r\times |\mathrm{pa}(v)|}$ and $c\in\mathbb{R}^r$, the Wald statistic satisfies
      \[
        W_n
        =
        n(C\hat\beta_v-c)^{\top}(C\hat{\mathcal{V}}_vC^{\top})^{-1}(C\hat\beta_v-c)
        \xrightarrow{d}\chi^2_r
        \quad\text{under }H_0.
      \]
      The matrix $C\hat{\mathcal{V}}_vC^{\top}$ is generically invertible because $C$ has full row rank.
\end{enumerate}}
\begin{proposition}[Confidence regions and Wald tests]
\label{prop:inference}
\TextPropInference
\end{proposition}
For the marginal test $H_0\colon\beta_{p_jv}=0$, the statistic reduces to $W_n=n\hat\beta_{p_jv}^2/\hat{\mathcal{V}}_v[j,j]\xrightarrow{d}\chi^2_1$; equivalently, $\hat\beta_{p_jv}/\sqrt{\hat{\mathcal{V}}_v[j,j]/n}$ is asymptotically standard normal.

\subsection{R~Implementation}

We implement the HTC estimator, asymptotic variance estimator, and inference in the \texttt{R} function \texttt{htcfit()}, with inferential output produced by \texttt{summary(htcfit())}.\footnote{Source code at \url{https://doi.org/10.5281/zenodo.21649507}; HTC order $\prec$ and witness sets are either determined via the \texttt{SEMID} package \citep{SEMID} or user-specified.} The output contains multiple blocks corresponding to each 
non-root node $v$ in the HTC order $\prec$. Each block is structured as follows.
\begin{description}[itemsep=3pt, parsep=0pt, leftmargin=1.5em, font=\normalfont\itshape]
\item[Block header.] Identifies $v$, its parents $\mathrm{pa}(v)$, and the witness set $Y_v$ with type labels \texttt{ext} (external witness, instrument $Z_y=X_y$) and \texttt{int} (internal witness, instrument $Z_y=\varepsilon_y$).
\item[Coefficient table.] For each edge $p_j\to v$: Reports the estimate $\hat\beta_{p_jv}$, the standard error $\smash[b]{\sqrt{\hat{\mathcal{V}}_v[j,j]/n}}$ from Proposition~\ref{prop:inference}(ii), the $z$-statistic $\hat\beta_{p_jv} / \smash[b]{\sqrt{\hat{\mathcal{V}}_v[j,j]/n}}$ which is asymptotically $N(0,1)$ distributed under $H_0\colon\beta_{p_jv}=0$, and a corresponding two-sided $p$-value.
\item[Footer.] Presents the empirical residual standard deviation $\hat\Omega_{vv}^{1/2} $, the structural $R^2 := 1 - \hat\Omega_{vv}/\mathrm{Var}(X_v)$, the fraction of $\mathrm{Var}(X_v)$ explained by $X_{\mathrm{pa}(v)}$
, and, when $|\mathrm{pa}(v)|>1$, the joint Wald statistic $n\hat\beta_v^\top \hat{\mathcal{V}}_v^{-1} \hat\beta_v \sim \chi^2_{|\mathrm{pa}(v)|}$ under $H_0\colon\beta_v = 0$ from Proposition~\ref{prop:inference}(iii).
\end{description}
\begin{example}
Figure~\ref{fig:summary_R} shows \texttt{summary(htcfit())} applied to the cyclic five-node graph on the left, with $n=1000$ Gaussian observations at the true parameters listed in Appendix~\ref{app:summary}.
The figure displays three of the four blocks in HTC order ($3\prec 5\prec 2 \prec 4$, node~$4$ truncated); nodes~$3$ and~$5$ each have a single parent, so their footers carry only the structural $R^2$, while node~$2$ adds the joint Wald $\chi^2$.
The block for node~$2$, which has two parents $\mathrm{pa}(2)=\{1,3\}$ identified by internal witnesses $Y_2=\{3,5\}\subset\mathrm{htr}(2)$, reports a joint Wald test that strongly rejects $H_0:\beta_{12}=\beta_{32}=0$.
The block for node~$5$ does not reject $H_0:\beta_{45}=0$, consistent with the true value $\beta_{45}=0$.
\end{example}

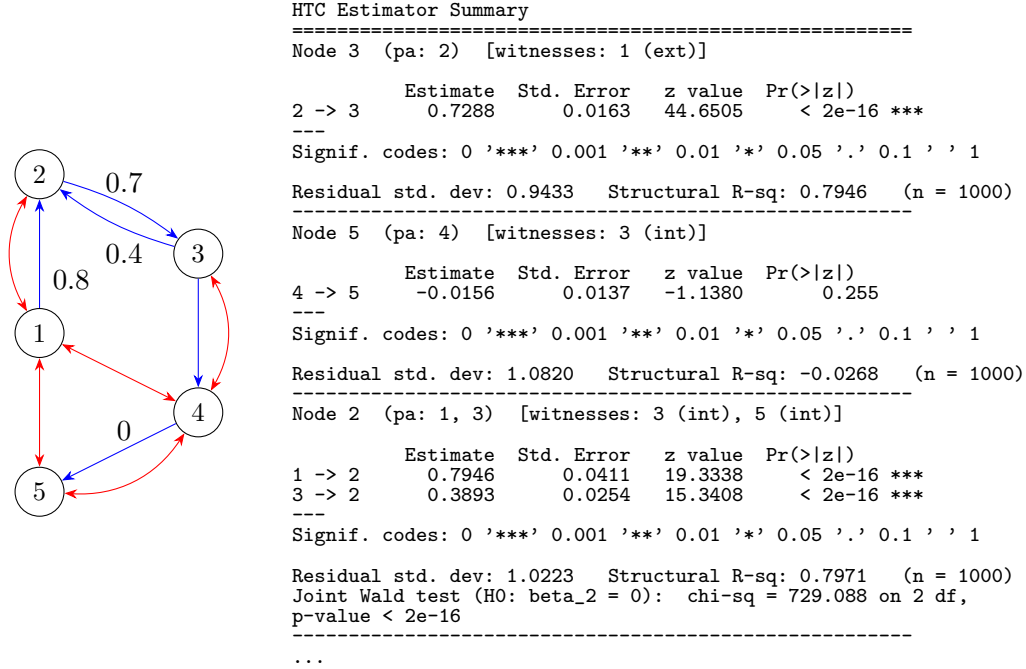
\begin{figure}
\begin{minipage}{0.2\textwidth}
\begin{tikzpicture}[scale = 0.7, >=Stealth, node style/.style={scale = 1, font=\small}]
\node[node style, draw, circle] (A) at (0,3) {$1$};
\node[node style, draw, circle] (B) at (0,6) {$2$};
\node[node style, draw, circle] (C) at (3,4.5) {$3$};
\node[node style, draw, circle] (D) at (3,1.5) {$4$};
\node[node style, draw, circle] (E) at (0,0) {$5$};
\node[black, scale = 1] (F) at (1.6,4.5) {$0.4$};
\node[black, scale = 1] (G) at (1.6,5.85) {$0.7$};
\node[black, scale = 1] (H) at (0.6,4) {$0.8$};
\node[black, scale = 1] (I) at (1.6,1.15) {$0$};

\def\dashlength{3pt}
\draw[blue, ->] (A) -- (B);
\draw[blue, ->] (B) to [bend left=10] (C);
\draw[blue, ->] (C) -- (D);
\draw[blue, ->] (D) -- (E);
\draw[blue, ->] (C) to [bend left=10] (B);
\draw[red, <->] (A) to [bend left=30] (B);
\draw[red, <->] (A) -- (D);
\draw[red, <->] (A) -- (E);
\draw[red, <->] (C) to [bend left=30] (D);
\draw[red, <->] (E) to [bend right=30] (D);
\end{tikzpicture}\end{minipage} \hfill
\begin{minipage}{0.75\textwidth}
\begin{Verbatim}[
  fontsize=\scriptsize,
  baselinestretch=0.8,
  xleftmargin=0pt,
  xrightmargin=0pt
]
HTC Estimator Summary
=======================================================
Node 3  (pa: 2)  [witnesses: 1 (ext)]

          Estimate  Std. Error   z value  Pr(>|z|)  
2 -> 3      0.7288      0.0163   44.6505     < 2e-16 ***
---
Signif. codes: 0 '***' 0.001 '**' 0.01 '*' 0.05 '.' 0.1 ' ' 1

Residual std. dev: 0.9433   Structural R-sq: 0.7946   (n = 1000)
-------------------------------------------------------
Node 5  (pa: 4)  [witnesses: 3 (int)]

          Estimate  Std. Error   z value  Pr(>|z|)  
4 -> 5     -0.0156      0.0137   -1.1380       0.255  
---
Signif. codes: 0 '***' 0.001 '**' 0.01 '*' 0.05 '.' 0.1 ' ' 1

Residual std. dev: 1.0820   Structural R-sq: -0.0268   (n = 1000)
-------------------------------------------------------
Node 2  (pa: 1, 3)  [witnesses: 3 (int), 5 (int)]

          Estimate  Std. Error   z value  Pr(>|z|)  
1 -> 2      0.7946      0.0411   19.3338     < 2e-16 ***
3 -> 2      0.3893      0.0254   15.3408     < 2e-16 ***
---
Signif. codes: 0 '***' 0.001 '**' 0.01 '*' 0.05 '.' 0.1 ' ' 1

Residual std. dev: 1.0223   Structural R-sq: 0.7971   (n = 1000)
Joint Wald test (H0: beta_2 = 0):  chi-sq = 729.088 on 2 df,  
p-value < 2e-16
-------------------------------------------------------
...
\end{Verbatim}
\end{minipage}
\caption{\texttt{summary(htcfit())} output for the cyclic five-node directed mixed graph (left, $n=1000$, true $\beta_{45}=0$); see Appendix~\ref{app:summary} for graph and parameter details.}
\label{fig:summary_R}
\end{figure}

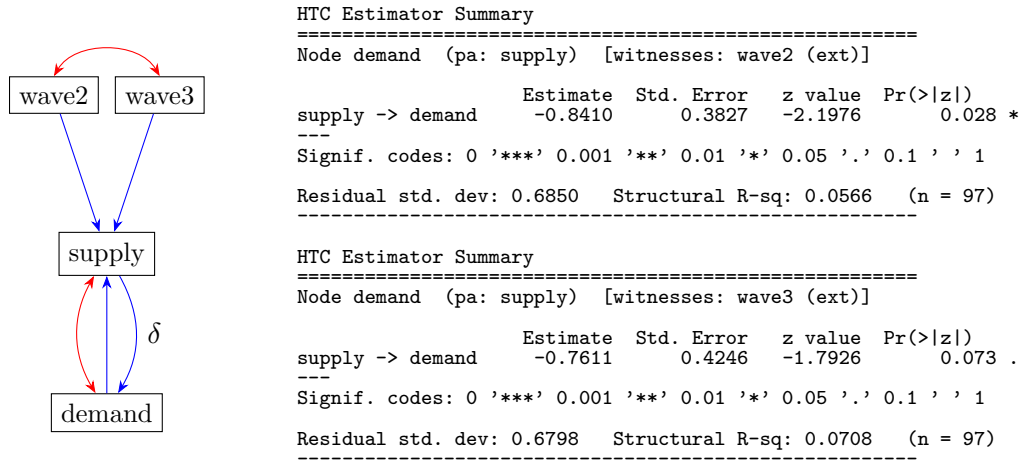
\begin{figure}
\centering
\begin{minipage}{0.2\textwidth}
\begin{tikzpicture}[scale = 0.7, >=Stealth, node style/.style={scale = 1, font=\small}]
\node[node style, draw] (A) at (0,3) {supply};
\node[node style, draw] (B) at (-1,6) {wave2};
\node[node style, draw] (C) at (1,6) {wave3};
\node[node style, draw] (E) at (0,0) {demand};
\node[black, scale = 1] (I) at (0.9,1.5) {$\delta$};

\def\dashlength{3pt}
\draw[blue, ->] (B) -- (A);
\draw[blue, ->] (C) -- (A);
\draw[red, <->] (A) to [bend right=30] (E);
\draw[red, <->] (B.north) to [bend left=60] (C.north);
\draw[blue, <-] (A) -- (E);
\draw[blue, ->] (A) to [bend left=30]  (E);
\end{tikzpicture}\end{minipage} \hfill
\begin{minipage}{0.75\textwidth}
\begin{Verbatim}[
  fontsize=\scriptsize,
  baselinestretch=0.8,
  xleftmargin=0pt,
  xrightmargin=0pt
]
HTC Estimator Summary
=======================================================
Node demand  (pa: supply)  [witnesses: wave2 (ext)]

                    Estimate  Std. Error   z value  Pr(>|z|)  
supply -> demand     -0.8410      0.3827   -2.1976       0.028 *
---
Signif. codes: 0 '***' 0.001 '**' 0.01 '*' 0.05 '.' 0.1 ' ' 1

Residual std. dev: 0.6850   Structural R-sq: 0.0566   (n = 97)
-------------------------------------------------------

HTC Estimator Summary
=======================================================
Node demand  (pa: supply)  [witnesses: wave3 (ext)]

                    Estimate  Std. Error   z value  Pr(>|z|)  
supply -> demand     -0.7611      0.4246   -1.7926       0.073 .
---
Signif. codes: 0 '***' 0.001 '**' 0.01 '*' 0.05 '.' 0.1 ' ' 1

Residual std. dev: 0.6798   Structural R-sq: 0.0708   (n = 97)
-------------------------------------------------------
\end{Verbatim}
\end{minipage}
\caption{\texttt{summary(htcfit())} applied to the \textit{Fulton Fish Market} dataset ($n=97$): cyclic directed mixed graph with correlated (\textit{supply}, \textit{demand}) errors (left), estimated separately with \textit{wave2} and \textit{wave3} as instruments for the demand elasticity~$\delta$.}
\label{fig:fish}
\end{figure}
\subsection{The Fulton Fish Market}
The Fulton Fish Market in New York City offers a textbook illustration of simultaneous-equations identification.
\citet{Angrist2000} analyze $n=97$ daily records of wholesale whiting transactions collected by \citet{Graddy1995}: log price (\textit{supply}), log quantity sold (\textit{demand}), as well as two- and three-day moving average wave height (\textit{wave2}, \textit{wave3}) off the Long Island coast.
Supply and demand are jointly determined, and the two structural equations share correlated errors from omitted common factors such as seasonal stock variation and port conditions.
The directed mixed graph in Figure~\ref{fig:fish} (left) captures this structure: directed edges \textit{supply}\,$\to$\,\textit{demand} and \textit{demand}\,$\to$\,\textit{supply} form the simultaneous-equations cycle, and a bidirected edge \textit{supply}\,$\leftrightarrow$\,\textit{demand} represents the correlated errors.
Rough seas reduce landings and raise prices without directly affecting consumer demand, making wave height a natural external instrument for the supply equation.
The target parameter is $\delta := \beta_{supply, \, demand}$, the demand elasticity. Naive OLS yields $\hat\delta_{\mathrm{OLS}}=-0.525$.

The HTC identifies $\delta$ using \textit{wave2} as the single external witness $Y_{\textit{demand}}=\{\textit{wave2}\}$.
Since \textit{wave2} is an external witness, no residual plug-in is required and the instrument is the raw variable $X_\textit{wave2}$.
Figure~\ref{fig:fish} (right) shows $\hat\delta=-0.8410$ (SE $0.3827$, $z=-2.1976$, $p=0.028$), a statistically significant negative demand elasticity directly comparable to the 
estimate of $-1.01$ (SE $0.42$) from \citet{Angrist2000}.
The estimates agree within one standard error, and the HTC asymptotic standard error is distribution-free, requiring only finite fourth moments.
The additional day of averaging in \textit{wave3} dilutes the sharp day-before supply shock and reduces the first-stage correlation with \textit{supply} from $0.4931$ to $0.3798$, making it a slightly weaker instrument for $\delta$: Figure~\ref{fig:fish} shows $\hat\delta=-0.7611$ (SE $0.4246$, $p=0.073$), within one standard error of the \textit{wave2} result.
Further details on data preprocessing and identification can be found in Appendix~\ref{app:fish}.

\section{Conclusion}
We have derived the semiparametric influence function $\phi_{\beta_v}=A_v^{-1}R_v$ of the HTC estimator for every HTC-identified node in an arbitrary directed mixed graph (Theorem~\ref{thm:htc_if}).
The theory covers causal structures beyond the classical IV setting: multiple parents, external and internal (residual-based) instruments, latent confounding, and cyclic graphs.
Under standard regularity conditions and finite fourth moments, the asymptotic variance estimator $\hat{\mathcal{V}}_v$ is computable in closed form by descent through the HTC order $\prec$ and is well calibrated across Gaussian and non-Gaussian error distributions. Applied to the Fulton Fish Market, the theory delivers a complete inferential summary for the demand elasticity in a simultaneous-equations model with correlated errors.

Two open problems emerge from this work.
First, the HTC algorithm as implemented in the \texttt{SEMID} package outputs one valid witness set and HTC ordering $\prec$ without variance guarantees, while Figure~\ref{fig:efficiency} shows this choice has a first-order effect on variance; a polynomial-time selection rule for a lower-variance $Y_v$ or $\prec$ has yet to be found.
Second, the HTC estimator is semiparametrically efficient in the classical single-instrument IV case, where $\phi_{\beta_{pv}} = X_y\varepsilon_v / \E{X_yX_p}$ coincides with the efficient influence function; for which directed mixed graphs efficiency holds within the HTC class has yet to be characterized.
Resolving these two problems would connect the graphical identification theory of \citet{Foygel2012} to a complete theory of variance-optimal HTC estimation.


\section*{Acknowledgments}

Leopold Mareis received funding from the German Research Foundation (DFG) under the Mathematical Research Data Initiative (project No. 460135501). 
Mathias Drton received funding from the European Research Council (ERC) under the European Union’s Horizon 2020 research and innovation programme (grant agreement No. 883818). 
We thank Benjamin Hollering for valuable discussions concerning marginal independence models.


\newpage
\bibliography{./bibliography}

\newpage
\appendix
\section{Assumptions}
\label{asm:regularity}

The following conditions hold at the true parameter $(\beta_0, \varepsilon_0) \in \mathcal{B} \times \mathcal{E}$.
\begin{enumerate}[label=(\roman*), itemsep=3pt, parsep=0pt, leftmargin=2em]
  \item \textit{(Hellinger differentiability.)} \label{asm:1}
        Every parametric path $\gamma\mapsto P_\gamma$ in $\mathcal{M}_{G}$ through
        $P_0$ is differentiable in quadratic mean at $\gamma=0$
        \citep[Definition~25.14]{Vaart_1998}: there exists a score $S\in L^2_0$ satisfying
        \[
          \int\!\left(\frac{dP_\gamma^{1/2} - dP_0^{1/2}}{\gamma}
                     -\frac{1}{2}S\,dP_0^{1/2}\right)^{\!2}
          \;\to\; 0
          \quad\text{as }\gamma\to 0.
        \]
        This is the standard regularity condition for RAL theory.
  \item \textit{(Non-degeneracy.)} \label{asm:2}
        The HTC matrix $A_v$ is invertible for every HTC-identified node $v\in V$.
        By \cite[Lemma~2]{Foygel2012}, $\det(A_v)$ is a nonzero polynomial in $(\beta,\Omega)$, so invertibility holds for all $(\beta,\Omega)$ outside a proper algebraic subset of the parameter space.
\end{enumerate}

\newpage
\section{Proofs} \label{appendix_proofs}
 \restatement{Lemma \ref{lem:htc_matrix}}{\TextLemmaHTCRelevanceMatrix}
\begin{proof}
Write $\Sigma:=\E{XX^\top}=(I-\beta)^{-\top}\Omega(I-\beta)^{-1}$.

\textbf{Case $y_i\in\mathrm{htr}(v)$.}
From $X=\beta^\top X+\varepsilon$, we have $Z_{y_i}=\varepsilon_{y_i}=(I-\beta^\top)_{y_i\cdot}X$, so $(A_v)_{ij}=\E{Z_{y_i}X_{p_j}}=[(I-\beta)^\top\Sigma]_{y_ip_j}$, a rational function of $(\beta,\Omega)$.

\textbf{Case $y_i\notin\mathrm{htr}(v)$.}
$Z_{y_i}=X_{y_i}$, so $(A_v)_{ij}=\Sigma_{y_ip_j}$, a rational function of $(\beta,\Omega)$.

Hence $\det(A_v)$ is rational in $(\beta,\Omega)$, not identically zero by \cite[Lemma~2]{Foygel2012}, so $A_v$ is invertible outside the proper algebraic subset $\{\det(A_v)=0\}\subset\mathcal{B}\times \mathrm{PD}(V)$.
\end{proof}

\restatement{Lemma \ref{lem:row_if}}{\TextRowMomentDerivative}
\begin{proof}
We prove the result by induction on the HTC ordering $\prec$.
The base nodes satisfy $Y_v\cap\mathrm{htr}(v)=\emptyset$, so all witnesses are external and the internal case does not arise; the base case therefore reduces to the external case below.
For any $a,b\in V$, differentiating $\Sigma_{ab}=\E{X_aX_b}$ along a parametric submodel with score $S$ gives $\tfrac{d}{d\gamma}\big|_0\Sigma_{ab,\gamma}=\E{(X_aX_b-\Sigma_{ab})S}$, so $\phi_{\Sigma_{ab}}=X_aX_b-\Sigma_{ab}$ is a valid mean-zero influence function for $\Sigma_{ab}$ in $\mathcal{M}_G$.

\textbf{Case $y\notin\mathrm{htr}(v)$.}
$M_y(t)=\E{X_yX_t}=\Sigma_{yt}$, so $\phi_{M_y(t)}=\phi_{\Sigma_{yt}}=X_yX_t-\Sigma_{yt}$.
This is mean-zero and satisfies $\E{\phi_{M_y(t)}S}=\tfrac{d}{d\gamma}\big|_0 M_{y,\gamma}(t)$ for all submodel scores $S$, hence it is a valid influence function for $M_y(t)$.

\textbf{Case $y\in\mathrm{htr}(v)$.}
Since $y\prec v$ in HTC order, the influence functions $\phi_{\beta_{qy}}$ for $q\in\mathrm{pa}(y)$ are available as valid mean-zero influence functions by the induction hypothesis.
Substituting $Z_y=X_y-\sum_{q\in\mathrm{pa}(y)}\beta_{qy}X_q$ into $M_y(t)=\E{Z_yX_t}$ gives
\[
  M_y(t) = \E{\varepsilon_yX_t} = \Sigma_{yt}-\sum_{q\in\mathrm{pa}(y)}\beta_{qy}\Sigma_{qt}.
\]
Differentiating $M_{y,\gamma}(t)=\Sigma_{yt,\gamma}-\sum_q\beta_{qy,\gamma}\Sigma_{qt,\gamma}$ along the submodel using the product rule for pathwise derivatives ($\phi_{fg}=\phi_f g_0+f_0\phi_g$, a direct consequence of the bilinearity of differentiation; see \citet[\S 25.7]{Vaart_1998}) and substituting $\phi_{\Sigma_{ab}}=X_aX_b-\Sigma_{ab}$:
\[
  \phi_{M_y(t)}
  =
  (X_yX_t-\Sigma_{yt})
  -\sum_{q\in\mathrm{pa}(y)}\phi_{\beta_{qy}}\,\Sigma_{qt}
  -\sum_{q\in\mathrm{pa}(y)}\beta_{qy}(X_qX_t-\Sigma_{qt}).
\]
Collecting terms: the leading random parts combine as $X_yX_t-\sum_q\beta_{qy}X_qX_t=\varepsilon_yX_t$, the constant parts reduce to $-\Sigma_{yt}+\sum_q\beta_{qy}\Sigma_{qt}=-\E{\varepsilon_yX_t}$, and the remaining terms yield $-\sum_{q\in\mathrm{pa}(y)}\Sigma_{qt}\,\phi_{\beta_{qy}}$, giving
\[
  \phi_{M_y(t)}=\varepsilon_yX_t-\E{\varepsilon_yX_t}-\sum_{q\in\mathrm{pa}(y)}\Sigma_{qt}\,\phi_{\beta_{qy}}.
\]
This formula is mean-zero: the first two terms are centered by construction, and each $\Sigma_{qt}\phi_{\beta_{qy}}$ is mean-zero because $\phi_{\beta_{qy}}$ is mean-zero by the induction hypothesis.
Since the formula was derived by differentiating $M_{y,\gamma}(t)$ along an arbitrary parametric submodel, it satisfies $\E{\phi_{M_y(t)}S}=\tfrac{d}{d\gamma}\big|_0 M_{y,\gamma}(t)$ for all submodel scores $S$; hence it is a valid influence function for $M_y(t)$.
\end{proof}

\restatement{Lemma \ref{lem:explicit_row}}{\TexRowContribution}
\begin{proof}
We substitute the formulas from Lemma~\ref{lem:row_if} into $R_{y,v}:=\phi_{M_y(v)}-\sum_{p\in\mathrm{pa}(v)}\beta_{pv}\phi_{M_y(p)}$ and simplify case by case.
Since each $\phi_{M_y(t)}$ is a valid mean-zero influence function by Lemma~\ref{lem:row_if}, and $R_{y,v}$ is a linear combination of these with fixed scalar coefficients $\beta_{pv}$, the row contribution $R_{y,v}$ is itself a valid mean-zero influence function for the row residual $M_y(v)-\sum_p\beta_{pv}M_y(p)$.

\textbf{Case $y\notin\mathrm{htr}(v)$.}
Lemma~\ref{lem:row_if} gives $\phi_{M_y(t)}=X_yX_t-\Sigma_{yt}$, so
\begin{align*}
  R_{y,v}
  &=
  X_yX_v-\Sigma_{yv}
  -\sum_{p\in\mathrm{pa}(v)}\beta_{pv}(X_yX_p-\Sigma_{yp}) \\
  &=
  X_y\!\Bigl(X_v-\sum_{p}\beta_{pv}X_p\Bigr)
  -\E{X_y\!\Bigl(X_v-\sum_{p}\beta_{pv}X_p\Bigr)}
  =
  X_y\varepsilon_v - \E{X_y\varepsilon_v}.
\end{align*}
The centering term vanishes: since $X=(I-\beta)^{-\top}\varepsilon$, we have $\E{X_y\varepsilon_v}=[(I-\beta)^{-\top}\Omega]_{yv}=\sum_k[(I-\beta)^{-1}]_{ky}\,\Omega_{kv}$.
Each summand is nonzero only when a directed path runs from $k$ to $y$ in $G$ \citep[Equation~2.2]{Foygel2012} and $\Omega_{kv}\neq 0$, i.e., $k\in\{v\}\cup\mathrm{sib}(v)$; such a pair constitutes a half-trek from $v$ to $y$, so every summand vanishes when $y\notin\mathrm{htr}(v)$.
Hence $R_{y,v}=X_y\varepsilon_v$, which is mean-zero.

\textbf{Case $y\in\mathrm{htr}(v)$.}
Lemma~\ref{lem:row_if} gives $\phi_{M_y(t)}=\varepsilon_yX_t-\E{\varepsilon_yX_t}-\sum_q\Sigma_{qt}\phi_{\beta_{qy}}$; substituting and expanding,
\begin{align*}
  R_{y,v}
  &=
  \varepsilon_yX_v - \E{\varepsilon_yX_v}
  -\sum_{q\in\mathrm{pa}(y)}\Sigma_{qv}\phi_{\beta_{qy}} \\
  &\quad
  -\sum_{p\in\mathrm{pa}(v)}\beta_{pv}\Bigl(
    \varepsilon_yX_p - \E{\varepsilon_yX_p}
    -\sum_{q\in\mathrm{pa}(y)}\Sigma_{qp}\phi_{\beta_{qy}}
  \Bigr).
\end{align*}
The coefficient of $\phi_{\beta_{qy}}$ in the expansion is $-\Sigma_{qv}+\sum_p\beta_{pv}\Sigma_{qp}$; since $$\E{X_q\varepsilon_v}=\E{X_q(X_v-\sum_p\beta_{pv}X_p)}=\Sigma_{qv}-\sum_p\beta_{pv}\Sigma_{qp},$$this equals $-\E{X_q\varepsilon_v}$.
The leading random terms give $\varepsilon_y(X_v-\sum_p\beta_{pv}X_p)=\varepsilon_y\varepsilon_v$; the centering terms give $-\E{\varepsilon_y\varepsilon_v}$, which equals zero because the HTC requires $Y_v\cap(\{v\}\cup\mathrm{sib}(v))=\emptyset$ (Definition~\ref{def:htc}), so $y\notin\mathrm{sib}(v)$, and the marginal independence model (Definition~\ref{def:parameter_space}) then gives $\varepsilon_y\perp\!\!\!\perp\varepsilon_v$, hence $\E{\varepsilon_y\varepsilon_v}=\Omega_{yv}=0$.
Assembling these three collections gives $R_{y,v}=\varepsilon_y\varepsilon_v-\sum_{q\in\mathrm{pa}(y)}\E{X_q\varepsilon_v}\,\phi_{\beta_{qy}}$, which is mean-zero: $\E{\varepsilon_y\varepsilon_v}=0$ by the argument above, and each $\phi_{\beta_{qy}}$ is mean-zero by Lemma~\ref{lem:row_if}.
\end{proof}

\restatement{Theorem \ref{thm:htc_if}}{\TextHTCInfluenceFunction}
\begin{proof}
We show that $\phi_{\beta_v}=A_v^{-1}R_v$ satisfies the three properties required of an influence function: it is mean-zero, square-integrable, and satisfies $\E{\phi_{\beta_v}S}=\dot\beta_v$ for every parametric submodel with score $S$ and $\dot\beta_v=\tfrac{d}{d\gamma}\big|_0\beta_{v,\gamma}$.

\textbf{Mean-zero and square-integrability.}
Since $R_{y,v}$ is mean-zero by Lemma~\ref{lem:explicit_row} for every $y\in Y_v$, we have $\E{R_v}=0$ and hence $\E{\phi_{\beta_v}}=0$.
Square-integrability $\phi_{\beta_v}\in L^2_0$ follows from $\varepsilon \in L^4_0$, which guarantees $R_{y,v}\in L^2_0$ for all $y,v$, and hence $\phi_{\beta_v}=A_v^{-1}R_v\in L^2_0$.

\textbf{Formula.}
Since $A_v$ is invertible at $P_0$ by Assumption \ref{asm:regularity}.\ref{asm:2}, the map $\beta_{v,\gamma}=A_{v,\gamma}^{-1}b_{v,\gamma}$ is differentiable in a neighbourhood of $P_0$, and the chain rule gives $\phi_{\beta_v}=A_v^{-1}\{\phi_{b_v}-(dA_v)\beta_v\}$.
The $i$th component of $\phi_{b_v}-(dA_v)\beta_v$ equals $\phi_{M_{y_i}(v)}-\sum_{p\in\mathrm{pa}(v)}\beta_{pv}\phi_{M_{y_i}(p)}=R_{y_i,v}$ by Lemma~\ref{lem:explicit_row}, so stacking over $i$ gives $\phi_{\beta_v}=A_v^{-1}R_v$.

\textbf{Score equation.}
Let $\gamma\mapsto P_\gamma$ be any parametric submodel in $\mathcal{M}_G$ with score $S$ and $\dot\beta_v=\tfrac{d}{d\gamma}\big|_0\beta_{v,\gamma}$.
We prove $\E{\phi_{\beta_v}S}=\dot\beta_v$ by induction on the HTC ordering $\prec$.

\emph{Base case.}
The base nodes satisfy $Y_v\cap\mathrm{htr}(v)=\emptyset$, so all witnesses are external, $Z_y=X_y$ is fixed, and Lemma~\ref{lem:explicit_row} gives $R_{y,v}=X_y\varepsilon_v$ with no correction terms.
The identification equation $\mathrm{E}_{P_\gamma}[X_{Y_v}\varepsilon_{v,\gamma}]=0$ holds for all $P_\gamma\in\mathcal{M}_G$; differentiating at $\gamma=0$ gives $\E{X_y\varepsilon_vS}+\E{X_yX_{\mathrm{pa}(v)}^\top}\dot\beta_v \cdot(-1)=0$, hence $\E{X_y\varepsilon_vS}=(A_v)_{y,\cdot}\dot\beta_v$ for each $y$.
Stacking gives $\E{R_vS}=A_v\dot\beta_v$ and therefore $\E{\phi_{\beta_v}S}=\dot\beta_v$.

\emph{Induction hypothesis.}
Assume $\E{\phi_{\beta_{qy}}S}=\dot\beta_{qy}$ for all $q\in\mathrm{pa}(y)$, $y\prec v$.

\emph{Inductive step.}
The identification equation $\mathrm{E}_{P_\gamma}[Z_{Y_v}(\beta_{\prec v,\gamma})\varepsilon_{v,\gamma}]=0$ holds at every $P_\gamma\in\mathcal{M}_G$, since $Y_v$ is a valid HTC witness set.
Differentiating at $\gamma=0$ via the product rule gives, for each witness $y\in Y_v$,
\[
  \E{\dot Z_y\,\varepsilon_v}+\E{Z_y\,\dot\varepsilon_v}+\E{Z_y\varepsilon_v S}=0.
\]

For $y\notin\mathrm{htr}(v)$: $Z_y=X_y$ is fixed, so $\dot Z_y=0$.
The structural equation $\varepsilon_v=X_v-\beta_v^\top X_{\mathrm{pa}(v)}$ gives $\dot\varepsilon_v=-X_{\mathrm{pa}(v)}^\top\dot\beta_v$, so $\E{Z_y\dot\varepsilon_v}=-(A_v)_{y,\cdot}\dot\beta_v$.
Rearranging: $\E{Z_y\varepsilon_vS}=(A_v)_{y,\cdot}\dot\beta_v$.

For $y\in\mathrm{htr}(v)$: $Z_y=X_y-\beta_y^\top X_{\mathrm{pa}(y)}$, so $\dot Z_y=-X_{\mathrm{pa}(y)}^\top\dot\beta_y$, giving the pathwise derivative $\E{\dot Z_y\varepsilon_v}=-\sum_{q\in\mathrm{pa}(y)}\E{X_q\varepsilon_v}\dot\beta_{qy}$.
As before, $\E{Z_y\dot\varepsilon_v}=-(A_v)_{y,\cdot}\dot\beta_v$.
Rearranging:
\[
  \E{Z_y\varepsilon_vS}=(A_v)_{y,\cdot}\dot\beta_v+\sum_{q\in\mathrm{pa}(y)}\E{X_q\varepsilon_v}\dot\beta_{qy}.
\]

For each internal witness, applying the induction hypothesis $\E{\phi_{\beta_{qy}}S}=\dot\beta_{qy}$ to the explicit formula $R_{y,v}=\varepsilon_y\varepsilon_v-\sum_q\E{X_q\varepsilon_v}\phi_{\beta_{qy}}$ from Lemma~\ref{lem:explicit_row} gives $\E{R_{y,v}S}=\E{Z_y\varepsilon_vS}-\sum_q\E{X_q\varepsilon_v}\dot\beta_{qy}$.
Substituting the differentiation result above, the correction terms cancel:
\[
  \E{R_{y,v}S}
  =\Bigl((A_v)_{y,\cdot}\dot\beta_v+\textstyle\sum_q\E{X_q\varepsilon_v}\dot\beta_{qy}\Bigr)
   -\textstyle\sum_q\E{X_q\varepsilon_v}\dot\beta_{qy}
  =(A_v)_{y,\cdot}\dot\beta_v.
\]
For external witnesses, $\E{R_{y,v}S}=\E{X_y\varepsilon_vS}=(A_v)_{y,\cdot}\dot\beta_v$ directly.
Stacking over all $y\in Y_v$ gives $\E{R_vS}=A_v\dot\beta_v$, and hence $\E{\phi_{\beta_v}S}=A_v^{-1}A_v\dot\beta_v=\dot\beta_v$.

In particular, for submodels varying only $\varepsilon \in \mathcal{E}$ with $\beta$ held fixed, $\dot\beta_y = 0$ for all $y \preceq v$, so the conclusion above gives $\mathbb{E}[\phi_{\beta_v} S] = 0$, confirming orthogonality to the nuisance tangent space of $\mathcal{M}_G$.
\end{proof}

\restatement{Proposition \ref{prop:recursive_variance}}{\TextPropRecursiveVariance}
\begin{proof}
The proof has two parts: we first establish asymptotic normality via the multi-stage Z-estimator expansion, then derive the three covariance formulas by substituting Lemma~\ref{lem:explicit_row} into $\E{R_{a,v}R_{b,v}}$.

At each node $v$ in HTC order, $\hat\beta_v$ solves $n^{-1}\sum_i Z_{Y_v}(\hat\beta_{\prec v})\varepsilon_v(X^{(i)},\hat\beta_v)=0$, where the instruments $Z_y=X_y-\hat\beta_y^\top X_{\mathrm{pa}(y)}$ for $y\in\mathrm{htr}(v)$ depend on the earlier-stage estimates $\hat\beta_{\prec v}$.
Applying Z-estimator theory stage by stage \citep[\S 5]{Vaart_1998}, the asymptotic expansion is $\sqrt{n}(\hat\beta_v-\beta_v)=n^{-1/2}\sum_i\phi_{\beta_v}(X^{(i)})+o_P(1)$, where the correction terms $\phi_{\beta_{qy}}$ in $R_v$ propagate the estimation uncertainty from $\hat\beta_{\prec v}$ into the instruments, matching exactly the influence function $\phi_{\beta_v}=A_v^{-1}R_v$ of Theorem~\ref{thm:htc_if}.
Since $\phi_{\beta_v}$ is mean-zero and $L^2_0$ by Theorem~\ref{thm:htc_if}, the central limit theorem gives $\sqrt{n}(\hat\beta_v-\beta_v)\xrightarrow{d}\mathcal{N}(0,\E{\phi_{\beta_v}\phi_{\beta_v}^\top})$.
Since $A_v=\E{Z_{Y_v}X_{\mathrm{pa}(v)}^\top}$ is a non-random population moment at $P_0$, substituting $\phi_{\beta_v}=A_v^{-1}R_v$ gives $\E{\phi_{\beta_v}\phi_{\beta_v}^\top}=A_v^{-1}\E{R_vR_v^\top}A_v^{-\top}=\mathcal{V}_v$.

It remains to derive the three formulas for $\E{R_{a,v}R_{b,v}}$.
In each case we substitute the formulas from Lemma~\ref{lem:explicit_row} and expand; the scalars $\E{X_q\varepsilon_v}$ are non-random population moments and factor out of joint expectations.
For $a,b\notin\mathrm{htr}(v)$: Lemma~\ref{lem:explicit_row} gives $R_{a,v}=X_a\varepsilon_v$ and $R_{b,v}=X_b\varepsilon_v$, so $\E{R_{a,v}R_{b,v}}=\E{X_aX_b\varepsilon_v^2}$.
For $a\in\mathrm{htr}(v)$, $b\notin\mathrm{htr}(v)$: Lemma~\ref{lem:explicit_row} gives $R_{a,v}=\varepsilon_a\varepsilon_v-\sum_{q\in\mathrm{pa}(a)}\E{X_q\varepsilon_v}\phi_{\beta_{qa}}$ and $R_{b,v}=X_b\varepsilon_v$; multiplying and factoring out the non-random scalar $\E{X_q\varepsilon_v}$:
\[
  \E{R_{a,v}R_{b,v}}=\E{\varepsilon_aX_b\varepsilon_v^2}-\sum_{q\in\mathrm{pa}(a)}\E{X_q\varepsilon_v}\,\E{\phi_{\beta_{qa}}X_b\varepsilon_v}.
\]
For $a,b\in\mathrm{htr}(v)$: both factors carry correction terms; substituting Lemma~\ref{lem:explicit_row} gives
\[
  R_{a,v}R_{b,v}=\Bigl(\varepsilon_a\varepsilon_v-\sum_{q\in\mathrm{pa}(a)}\E{X_q\varepsilon_v}\phi_{\beta_{qa}}\Bigr)\Bigl(\varepsilon_b\varepsilon_v-\sum_{r\in\mathrm{pa}(b)}\E{X_r\varepsilon_v}\phi_{\beta_{rb}}\Bigr).
\]
Expanding bilinearly and factoring out the non-random scalars $\E{X_q\varepsilon_v}$ and $\E{X_r\varepsilon_v}$ from each cross-term yields the four-term formula in the proposition.
\end{proof}

\restatement{Lemma \ref{lem:variance_recursion}}{\TextLemmaVariancerecursion}
\begin{proof}
The proof establishes the three claims in the lemma in order: the reduction formula (Claim 1), the case expansion (Claim 2), and termination of the recursion (Claim 3), from which computability follows.

\textbf{Claim 1 — Reduction formula.}
By Theorem~\ref{thm:htc_if}, the vector influence function for $\beta_a$ is $\phi_{\beta_a}=A_a^{-1}R_a\in\mathbb{R}^{|\mathrm{pa}(a)|}$; the scalar component for the edge $q=p_j\to a$ is its $j$th entry, $\phi_{\beta_{qa}}=e_j^\top A_a^{-1}R_a$.
Since $A_a=\E{Z_{Y_a}X_{\mathrm{pa}(a)}^\top}$ is a non-random population moment at $P_0$, it factors outside the expectation:
\[
  \E{\phi_{\beta_{qa}}X_b\varepsilon_v}
  =\E{e_j^\top A_a^{-1}R_aX_b\varepsilon_v}
  =e_j^\top A_a^{-1}\E{R_aX_b\varepsilon_v}.
\]

\textbf{Claim 2 — Case expansion.}
It remains to expand $\E{R_{y,a}X_b\varepsilon_v}$ for each $y\in Y_a$ by substituting Lemma~\ref{lem:explicit_row} applied at node $a$.
For $y\notin\mathrm{htr}(a)$: Lemma~\ref{lem:explicit_row} gives $R_{y,a}=X_y\varepsilon_a$, so $\E{R_{y,a}X_b\varepsilon_v}=\E{X_y\varepsilon_aX_b\varepsilon_v}$, a raw fourth-order moment.
For $y\in\mathrm{htr}(a)$: Lemma~\ref{lem:explicit_row} gives $R_{y,a}=\varepsilon_y\varepsilon_a-\sum_{q'\in\mathrm{pa}(y)}\E{X_{q'}\varepsilon_a}\phi_{\beta_{q'y}}$, where the scalars $\E{X_{q'}\varepsilon_a}$ are non-random population moments (note: the subscript is $\varepsilon_a$, arising from applying Lemma~\ref{lem:explicit_row} at node $a$, not $\varepsilon_v$).
Factoring these scalars outside the expectation gives
\[
  \E{R_{y,a}X_b\varepsilon_v}
  =\E{\varepsilon_y\varepsilon_aX_b\varepsilon_v}
  -\sum_{q'\in\mathrm{pa}(y)}\E{X_{q'}\varepsilon_a}\,\E{\phi_{\beta_{q'y}}X_b\varepsilon_v},
\]
which introduces new cross-variances $\E{\phi_{\beta_{q'y}}X_b\varepsilon_v}$ at the earlier node $y\prec a$.

\textbf{Claim 3 — Termination and computability.}
Applying Claim 1 to each new cross-variance $\E{\phi_{\beta_{q'y}}X_b\varepsilon_v}$ reduces it to cross-variances $\E{R_{y',y}X_b\varepsilon_v}$ at node $y\prec a$; the internal case of Claim 2 then introduces cross-variances at nodes strictly below $y$ in HTC order, and so on.
Since $\prec$ is a strict partial order on the finite index set $V$, this descent terminates; at base nodes $u$ with $Y_u\cap\mathrm{htr}(u)=\emptyset$, all witnesses are external, and the case expansion gives $\E{R_{y,u}X_b\varepsilon_v}=\E{X_y\varepsilon_uX_b\varepsilon_v}$, a pure fourth-order moment of observables requiring no further recursion.
Since $\E{R_vR_v^\top}$ is expressed through Proposition~\ref{prop:recursive_variance} in terms of finitely many cross-variances, and each reduces in finitely many steps to such fourth-order moments, $\E{R_vR_v^\top}$ is computable in finitely many steps.
\end{proof}

\restatement{Proposition \ref{prop:inference}}{\TextPropInference}
\begin{proof}
The finite fourth moments $\varepsilon \in L^4_0$ ensure that every term entering $\hat{\mathcal{V}}_v$ is the empirical average of an $L^2_0$-integrable function; by the law of large numbers, $\hat{\mathcal{V}}_v\xrightarrow{p}\mathcal{V}_v$.

\textit{(i)} Since $\sqrt{n}(\hat\beta_v-\beta_v)\xrightarrow{d}\mathcal{N}(0,\mathcal{V}_v)$ by Proposition~\ref{prop:recursive_variance}, the continuous mapping theorem gives $n(\hat\beta_v-\beta_v)^\top\mathcal{V}_v^{-1}(\hat\beta_v-\beta_v)\xrightarrow{d}\chi^2_{|\mathrm{pa}(v)|}$.
By Slutsky's theorem and $\hat{\mathcal{V}}_v\xrightarrow{p}\mathcal{V}_v$, replacing $\mathcal{V}_v$ by $\hat{\mathcal{V}}_v$ does not change the limit, giving the confidence ellipsoid.

\textit{(ii)} By the continuous mapping theorem (projection to coordinate $j$), $\sqrt{n}(\hat\beta_{p_jv}-\beta_{p_jv})\xrightarrow{d}\mathcal{N}(0,\mathcal{V}_v[j,j])$.
Since $\hat{\mathcal{V}}_v[j,j]\xrightarrow{p}\mathcal{V}_v[j,j]>0$, applying Slutsky's theorem gives $\sqrt{n}(\hat\beta_{p_jv}-\beta_{p_jv})/\sqrt{\hat{\mathcal{V}}_v[j,j]}\xrightarrow{d}\mathcal{N}(0,1)$, yielding the stated interval.

\textit{(iii)} Under $H_0\colon C\beta_v=c$, the continuous mapping theorem gives $\sqrt{n}(C\hat\beta_v-c)\xrightarrow{d}\mathcal{N}(0,C\mathcal{V}_vC^\top)$.
By Slutsky's theorem and the continuous mapping theorem (quadratic form), replacing $C\mathcal{V}_vC^\top$ by $C\hat{\mathcal{V}}_vC^\top$ gives $W_n\xrightarrow{d}\chi^2_r$.
Since $\mathcal{V}_v\succ 0$ generically under Assumption~\ref{asm:regularity}.\ref{asm:1} and $C$ has full row rank $r$, $C\mathcal{V}_vC^\top$ is positive definite; by consistency of $\hat{\mathcal{V}}_v$, $C\hat{\mathcal{V}}_vC^\top$ is almost surely invertible for large $n$.
\end{proof}

\newpage
\section{Simulation Study: Setup and Experiments} 
\subsection{Calibration Example} \label{app:calibration}

The graph in Figure~\ref{fig:G1} has directed edges $1\to 2$, $2\to 4$, $1\to 5$, $3\to 5$ and bidirected edges $1\leftrightarrow 3$, $1\leftrightarrow 4$, $1\leftrightarrow 5$.
The HTC ordering is $2\prec 4\prec 5$.
The true parameter matrix has non-zero entries
\[
  \beta_{12}=0.8,\quad\beta_{24}=0.7,\quad\beta_{15}=0.8,\quad\beta_{35}=0.6.
\]
The error covariance matrix $\Omega=\E{\varepsilon\varepsilon^\top}$, consistent with the marginal independence constraints $\varepsilon_w\indep\varepsilon_{V\setminus(\{w\}\cup \, \mathrm{sib}(w))}$ for each $w\in V$, has non-zero off-diagonal entries $\Omega_{13}=0.3$, $\Omega_{14}=0.75$, $\Omega_{15}=0.2$ only; all other off-diagonal entries are zero.
Positive definiteness holds since $\Omega_{13}^2+\Omega_{14}^2+\Omega_{15}^2=0.6525<1$.

Two error distributions are studied.
The Gaussian distribution samples $\varepsilon\sim\mathcal{N}_5(0,\Omega)$ directly.
The non-Gaussian distribution constructs errors $\varepsilon$ from five independent centered Gamma distributions $\varepsilon_i'\sim(\mathrm{Gamma}(2,1)-2)/\!\sqrt{2}$ via the Cholesky factorization of the $\{1,3,4,5\}$ block of $\Omega$ (node $2$ is independent of all others).
Both distributions match the error covariance $\Omega$ exactly.

\textbf{Experiment to Figure~\ref{fig:calibration}.}
We created $100$ replications of datasets with sample sizes $n\in\{250,500,1000,2000,4000,8000\}$ for each of the two error distributions.
For each dataset, the recursive HTC estimator $\hat\beta_5=(\hat\beta_{15},\hat\beta_{35})^\top$ and the estimated asymptotic covariance matrix $\hat{\mathcal{V}}_5$ were computed as in Proposition~\ref{prop:recursive_variance}.
The four panels of the combined figure compare empirical and estimated quantities across replications:
\begin{enumerate}[label=(\roman*), itemsep=3pt, parsep=0pt, leftmargin=2em]
  \item $\mathrm{Var}(\hat\beta_{15})$ against $\hat{\mathcal{V}}_5[1,1]/n$,
  \item $\mathrm{Var}(\hat\beta_{35})$ against $\hat{\mathcal{V}}_5[2,2]/n$,
  \item $|\mathrm{Cov}(\hat\beta_{15},\hat\beta_{35})|$ against $|\hat{\mathcal{V}}_5[1,2]|/n$ (absolute value, since the covariance is negative; log scale),
  \item $\mathrm{Cor}(\hat\beta_{15},\hat\beta_{35})$ against $\hat{\mathcal{V}}_5[1,2]/\sqrt{\hat{\mathcal{V}}_5[1,1]\hat{\mathcal{V}}_5[2,2]}$ (linear scale).
\end{enumerate}
Estimated quantities are averaged over the $100$ replications.
The correlation panel (iv) is included to demonstrate that the estimators $\hat\beta_{15}$ and $\hat\beta_{35}$ are asymptotically negatively correlated, a consequence of the internal witness $4\in\mathrm{htr}(5)$ introducing a shared correction term in $R_{4,5}$.

\textbf{Experiment to Figure~\ref{fig:efficiency}.}
We created $100$ replications of datasets with sample sizes $n\in\{250,500,1000,2000,4000,8000\}$ for each of the two error distributions.
Three valid HTC witness sets for node $2$ are compared: $Y_2\in\bigl\{\{1\},\{3\},\{5\}\bigr\}$.
Since $\mathrm{htr}(2)=\{4\}$, every witness is external and $Z_y=X_y$.
For each dataset and witness set, the recursive HTC estimator $\hat\beta_{12}$ and the estimated variance $\hat{\mathcal{V}}_2[1,1]/n$ were computed as in Proposition~\ref{prop:recursive_variance}.
Each of the three panels plots $\mathrm{Var}(\hat\beta_{12})$ against the mean of $\hat{\mathcal{V}}_2[1,1]/n$ over the $100$ replications, on a shared log-scale $y$-axis.

\subsection{Inference Example} \label{app:summary}

The five-node graph of Figure~\ref{fig:summary_R} has directed edges $1\to 2$, $2\to 3$, $3\to 2$, $3\to 4$, $4\to 5$ and bidirected edges $1\leftrightarrow 2$, $1\leftrightarrow 4$, $1\leftrightarrow 5$, $3\leftrightarrow 4$, $4\leftrightarrow 5$.
The true parameter matrix has non-zero entries
\[
  \beta_{12}=0.8,\quad\beta_{23}=0.7,\quad\beta_{32}=0.4,\quad\beta_{34}=0.8,\quad\beta_{45}=0,
\]
with $\beta_{45}=0$ included to illustrate inference at a zero coefficient.
The directed cycle $2\to 3\to 2$ is admissible since $\det(I-\beta)\neq 0$ at the true parameter.
The HTC ordering is $3\prec 5\prec 2\prec 4$: node~$3$ is identified first via the external witness $Y_3=\{1\}$; node~$5$ follows with the internal witness $Y_5= \{3\}$; node~$2$ is identified using the internal witnesses $Y_2=\{3,5\}\subset\mathrm{htr}(2)$; node $4$ is identified last through the internal witness $Y_4 = \{2\}$.
The error covariance matrix $\Omega=\E{\varepsilon\varepsilon^\top}$, consistent with the marginal independence constraints $\varepsilon_w\indep\varepsilon_{V\setminus(\{w\}\cup\, \mathrm{sib}(w))}$ for each $w\in V$, has non-zero off-diagonal entries $\Omega_{12}=0.5$, $\Omega_{14}=0.25$, $\Omega_{15}=0.75$, $\Omega_{34}=0.5$, $\Omega_{45}=0.4$; all remaining off-diagonal entries are zero and all diagonal entries equal one.
Errors are drawn as $\varepsilon\sim\mathcal{N}_5(0,\Omega)$ and observables are computed via $X=(I-\beta)^{-\top}\varepsilon$ for a single dataset of size $n=1000$.

\section{Application: Fulton Fish Market}
\label{app:fish}

The dataset of \citet{Angrist2000} records 97 daily observations from December 1991 to May 1992, with three variables: log quantity-weighted average price (\texttt{lavgprc}), log total quantity sold (\texttt{ltotqty}), and wave height off the Long Island coast (\texttt{wave2} and \texttt{wave3}).
Following \citet{Angrist2000}, day-of-week effects (Monday--Thursday dummies, Friday as baseline) are partialed out from all three variables by regression, yielding mean-zero working residuals that remove the systematic day-of-week trading patterns and make the i.i.d.\ approximation underlying the asymptotic variance estimator more credible.

In the directed mixed graph, $\mathrm{pa}({demand})=\{{supply}\}$ and $\mathrm{sib}({demand})=\{{supply}\}$, so admissible HTC witnesses must lie outside $\{{supply},{demand}\}$.
Both \textit{wave2} and \textit{wave3} satisfy this condition for \textit{supply}\,$\to$\,\textit{demand} and are valid external witnesses, so the instruments are the raw variables $Z_{\textit{wave2}}=X_{\textit{wave2}}$ and $Z_{\textit{wave3}}=X_{\textit{wave3}}$, respectively, with first-stage correlations \[\mathrm{Cor}(\textit{wave2},\textit{supply})=0.4931 \quad \text{and} \quad \mathrm{Cor}(\textit{wave3},\textit{supply})=0.3798.\]
The 3-day average smooths the sharp day-before supply shock and dilutes the first-stage signal, widening the asymptotic standard error from $0.3827$ to $0.4246$.
The estimates $\hat\delta=-0.8410$ (\textit{wave2}) and $\hat\delta=-0.7611$ (\textit{wave3}) lie within one standard error of each other, confirming that both instruments identify the same structural parameter.

The supply slope (coefficient on \textit{demand}\,$\to$\,\textit{supply}) is not identified, as no demand-side instrument is available in this dataset.

\end{document}